\documentclass[journal,twoside,web]{ieeecolor}
\usepackage{generic}
\usepackage{cite}
\usepackage{amsmath,amssymb,amsfonts}
\usepackage{algorithmic}
\usepackage{graphicx}
\usepackage{textcomp}
\usepackage{epstopdf}
\usepackage{mathrsfs}
\usepackage{pifont}

\newcommand{\beq}{\begin{equation}}
\newcommand{\eeq}{\end{equation}}
\newcommand{\bea}{\begin{eqnarray}}
\newcommand{\eea}{\end{eqnarray}}

\usepackage{enumerate}

\usepackage{float}

\usepackage{tikz}
\usetikzlibrary{matrix,arrows,decorations.pathmorphing}
\usepackage{verbatim}
\usepackage{mathtools}
\usepackage{caption}
\usepackage{subcaption}
\usetikzlibrary{arrows}
\usepackage{lipsum}
\usepackage{tabularx}
\newtheorem{assumption}{Assumption}
\newtheorem{cor}{Corollary}
\usepackage{amsmath,esint} 
\usepackage{cuted}
\DeclareMathOperator*{\sgn}{\mathsf{sgn}}
\DeclareMathOperator*{\Dg}{\mathsf{diag}}

\def\BibTeX{{\rm B\kern-.05em{\sc i\kern-.025em b}\kern-.08em
    T\kern-.1667em\lower.7ex\hbox{E}\kern-.125emX}}
\begin{document}
\title{Higher-order Laplacian dynamics on hypergraphs with cooperative and antagonistic interactions}
\author{Shaoxuan Cui, Chencheng Zhang, 
        Bin Jiang,  \IEEEmembership{Fellow, IEEE,}, Hildeberto Jardón-Kojakhmetov\\
        and~Ming Cao, \IEEEmembership{Fellow, IEEE}
\thanks{Cui was supported by the China Scholarship Council. This work of Zhang and Jiang was supported in part by the National Natural Science Foundation of China under Grants 62303218 and 62403241. Zhang was supported in addition by the Postdoctoral Fellowship Program of CPSF under Grant GZB20240975 and Jiangsu Provincial Outstanding Postdoctoral Program. The work of Cao was supported in part by the Netherlands Organization for Scientific Research (NWO-Vici-19902). }
\thanks{S. Cui and H. Jard\'on-Kojakhmetov are with the Bernoulli Institute for Mathematics, Computer Science and Artificial Intelligence, University of Groningen, Groningen, 9747 AG Netherlands.   with the Bernoulli Institute for Mathematics, Computer Science and Artificial Intelligence, University of Groningen, Groningen, 9747 AG Netherlands  (E-mail:  s.cui@rug.nl; h.jardon.kojakhmetov@rug.nl)}
\thanks{C. Zhang and B. Jiang are with the College of Automation Engineering, Nanjing University of Aeronautics and Astronautics, Nanjing, 211106, China (e-mail: zhangchencheng@nuaa.edu.cn;  binjiang@nuaa.edu.cn). 
(Corresponding author: Chencheng Zhang)}
\thanks{M. Cao is with the Institute of Engineering and Technology (ENTEG), University of Groningen, Groningen, 9747 AG Netherlands  (e-mail: m.cao@rug.nl).
}}

\maketitle

\newtheorem{remark}{Remark}
\newtheorem{lemma}{Lemma}
\newtheorem{thm}{Theorem}
\newtheorem{example}{Example}
\newtheorem{definition}{Definition}
\newtheorem{prop}{Proposition}

\begin{abstract}
Laplacian dynamics on a signless graph characterize a class of linear interactions, where pairwise cooperative interactions between all agents lead to the convergence to a common state. On a structurally balanced signed graph, the agents converge to values of the same magnitude but opposite signs (bipartite consensus), as illustrated by the well-known Altafini model. These interactions have been modeled using traditional graphs, where the relationships between agents are always pairwise. In comparison, higher-order networks (such as hypergraphs), offer the possibility to capture more complex, group-wise interactions among agents. This raises a natural question: can collective behavior be analyzed by using hypergraphs? The answer is affirmative. In this paper, higher-order Laplacian dynamics on signless hypergraphs are first introduced and various collective convergence behaviors are investigated, in the framework of homogeneous and non-homogeneous polynomial systems. Furthermore, by employing gauge transformations and leveraging tensor similarities, we extend these dynamics to signed hypergraphs, drawing parallels to the Altafini model. Moreover, we explore non-polynomial interaction functions within this framework. The theoretical results are demonstrated through several numerical examples.
\end{abstract}


\section{Introduction}
\label{sec:introduction}
\IEEEPARstart{F}{locking}
The study of Laplacian dynamics has long served as a fundamental behavior framework for networks or multi-agent systems, especially in understanding how interactions between agents lead to converging behaviors\cite{turner1957collective}.  
In addition, Laplacian dynamics have attracted the attention of researchers among different disciplines due to their potential application in  biology \cite{reynolds1987flocks,vicsek1995novel}, social sciences \cite{rainer2002opinion,vasca2021practical} and control engineering \cite{lin2004local,zhao2020consensus,zhao2021further,saber2003consensus}. 
Particularly relevant for this work are, for example, \cite{lin2015graph, zhou2009convergence, xia2017analysis, lin2013leader, tian2023laplacian} that have focused on the convergence of systems through cooperative interactions.


Networks are generally represented using either signless or signed graphs.
In a signless graph where the weight of each edge is nonnegative, Laplacian dynamics typically describe cooperative interactions, leading to a state where all agents converge to a common identical value. However, in a signed graph, where edges can have both positive and negative weights, Laplacian dynamics result in more complex collective behaviors \cite{shi2019dynamics, pan2016laplacian, meng2018uniform, wu2023uniform}.
Researchers have taken antagonistic interactions into account and showed convergence behavior on a signed network. One of the most well-known models in this context is Altafini's model \cite{altafini2012consensus}, which shows that all agents reach bipartite consensus; to be more precise, all agents converge to a value the same in modulus but different in signs if and only if the underlying signed network is structurally balanced. Otherwise, all agents converge to zero. Developing upon Altafini's model,  \cite{liu2017exponential,xia2015structural,meng2016behaviors,proskurnikov2014consensus}, have further extended it, in discrete time and continuous time, respectively, into the time-varying digraph case. While Altafini's model adopts the definition of Laplacian $L=[\ell_{ik}]$ as:
\begin{equation*}
\ell_{i k}= \begin{cases}\sum_{j \in \operatorname{adj}(i)}\left|a_{i j}\right| & k=i \\ -a_{i k} & k \neq i\end{cases};
\end{equation*}
some other researchers let $\ell_{ii}=\sum_{j \in \operatorname{adj}(i)}a_{i j}$ instead. It has been shown that the latter definition leads to more complicated system behaviors (convergence, clustering, divergence) \cite{pan2016laplacian}.


So far, all the aforementioned works are based on a conventional network and all interactions between agents are pairwise. However, in reality, the interactions are not always pairwise. For example, in social networks, people communicate with each other not only via one-on-one private chat but also via group chat \cite{cui2023general}. In ecology, the influence of one species on another is usually correlated to a third species \cite{letten2019mechanistic,cui2024analysis}. All these facts lead to a need for higher-order interactions (HOIs), which can be captured by higher-order networks \cite{bick2021higher} (simplicial complexes or hypergraphs). 

Recently, more and more works concentrate on how the consensus behavior emerges on a hypergraph and how to design a consensus protocol over a signless higher-order network \cite{sahasrabuddhe2021modelling,neuhauser2021consensus,neuhauser2022consensus,neuhauser2020multibody,neuhauser2021opinion,shang2022system}. All such works adopt a model in the form of
\begin{equation}\label{eq:l0}
\dot{x}_i=\sum_{j, k=1}^N A_{i j k} \ s\left(\left|x_j-x_k\right|\right)\left[\left(x_j-x_i\right)+\left(x_k-x_i\right)\right],
\end{equation}
where $A_{ijk}$ are entries of the adjacency tensor of a hypergraph.
However, the nonlinearity of the system mainly comes from the scaling function $s\left(\left|x_j-x_k\right|\right)$ rather than the higher-order network, and so, the effect of higher-order interactions has not been fully revealed. In essence, the dynamics of \eqref{eq:l0} are not governed by a higher-order Laplacian structure, which will be properly defined and introduced later in this paper (section \ref{sec:laplacian}). Furthermore, there is a gap in understanding the potential (bipartite) collective behavior on signed higher-order networks. In other words, what types of higher-order structures can lead to bipartite convergence is still unclear.

In the context of conventional graphs, linear dynamics are typically associated with the Laplacian matrix, where the system takes the form $\Dot{x}=-Lx$. 
Regarding higher-order interactions, the concept of a Laplacian tensor of a uniform hypergraph is proposed by, for example, \cite{hu2015laplacian,qi2013h}. They show that a Laplacian tensor shares many useful properties similar to a Laplacian matrix. Thus, the natural way to investigate the collective behavior and convergence property of a higher-order network on hypergraphs is to consider the system in the higher-order Laplacian form of $\Dot{x}=-Lx^{k-1}$, where $L$ is a Laplacian tensor of order $k$ (we will introduce the concept of hypergraph Laplacian later in section \ref{sec:laplacian} of this paper).

The contributions of our paper are summarized as follows. First, we investigate a class of higher-order networks on hypergraphs where the corresponding tensor is an irreducible Metzler tensor with a zero Perron-eigenvalue. This result is then utilized to analyze the convergence property of the system with the help of a Laplacian tensor on a signless hypergraph. 
The concept of a Laplacian tensor originates from the study of undirected, unweighted uniform hypergraphs. We extend the definition of an undirected, unweighted Laplacian tensor to a weighted directed hypergraph. By exploiting the properties of similar tensors and a gauge transformation, we finally investigate the collective behavior of higher-order Laplacian dynamics over a signed hypergraph and further utilize it to design a collective protocol on a signed hypergraph that ensures bipartite convergence. For the first time, we extend the concept of structural balance into the hypergraph context. Moreover, we leverage our results and consider also non-polynomial interaction functions for higher-order Laplacian dynamics. Finally, numerical examples are given to highlight the theoretical contributions.

\emph{Notation:} $\mathbb{R}$ denotes the set of real numbers, $\mathbb{R}_{+}$ denotes the set of nonnegative real numbers and $\mathbb{R}_{++}$ denotes the set of positive real numbers. Given a square matrix $M \in \mathbb{R}^{n \times n}$, $\rho(M)$ denotes the spectral radius of $M$, which is the largest absolute value of the eigenvalues of $M$. For a matrix $M \in \mathbb{R}^{n \times r}$ and a vector $a \in \mathbb{R}^n$, $M_{ij}$ and $a_{i}$ denote the element in the $i$th row and $j$th column and the $i$th entry, respectively. For any two vectors $a, b \in \mathbb{R}^n$, $a > (<) b$ indicates that $a_i >(<) b_i$, for all $i=1,\ldots,n$; and $a \geq (\leq) b$ means that $a_i \geq (\leq) b_i$, for all $i=1,\ldots,n$. These component-wise comparisons are also valid for matrices or tensors with the same dimension. The vector $\mathbf{1}$ ($\mathbf{0}$) represents the column vector, matrix, or tensor of all ones (zeros) with appropriate dimensions. In the following section, we introduce the preliminaries and further notations regarding tensors. 
 
\section{Preliminaries on tensors and hypergraphs}

A tensor $T \in \mathbb{R}^{n_1 \times n_2 \times \cdots \times n_k}$ is a multidimensional array. The order of a tensor is $k$, which represents the number of its dimensions, and each dimension $n_i$, where $i = 1, \ldots, k$, is a mode of the tensor. If all the modes of a tensor have the same size, we call it a cubical tensor, denoted as $T \in \mathbb{R}^{n \times n \times \cdots \times n}$. 
A cubical tensor $T$ is supersymmetric if $T_{j_1 j_2 \ldots j_k}$ is invariant under any permutation of its indices.

The identity tensor of order $k$ dimension $n$, is defined by
\begin{equation*}
   \delta_{i_1 \cdots i_k} = \begin{cases}
   1 & \text{if } i_1 = i_2 = \cdots = i_k \\
   0 & \text{otherwise}
   \end{cases}.
\end{equation*}

Throughout the rest of this paper, the term 'tensor' will always refer to a cubical tensor unless explicitly stated otherwise. For a tensor, we consider the following homogeneous polynomial equation:


\begin{equation}\label{eq:eigenproblem}
A x^{k-1} = \lambda x^{[k-1]},
\end{equation}
where $A x^{k-1}$ and $x^{[k-1]}$ are vectors, and their $i$-th components are defined as follows:

\begin{align*}
\left(A x^{k-1}\right)_i & = \sum_{i_2, \ldots, i_k=1}^n A_{i, i_2 \cdots i_k} x_{i_2} \cdots x_{i_k}, \\
\left(x^{[k-1]}\right)_i & = x_i^{k-1}.
\end{align*}

If there exists a real number $\lambda$ and a nonzero real vector $x$ that are solutions of \eqref{eq:eigenproblem}, then $\lambda$ is referred to as the $H$-eigenvalue of $A$, and $x$ is called the $H$-eigenvector of $A$ associated with $\lambda$. We further classify $\lambda$ and $x$ as $H^{++}$-eigenvalue and $H^{++}$-eigenvector, respectively, if $x$ is strictly greater than the zero vector. Throughout this paper, we use eigenvalues and eigenvectors as well as $H$-eigenvalues and $H$-eigenvectors interchangeably.


It is straightforward to see that if $Ax^{k-1}=k_1 x^{[k-1]}$ and $Bx^{k-1}=k_2 x^{[k-1]}$, then $Ax^{k-1}+Bx^{k-1}=(A+B)x^{k-1}=k_1 x^{[k-1]}+k_2 x^{[k-1]}=(k_1+k_2)x^{[k-1]}$. We will utilize these properties later in the technical part. 

The spectral radius of a tensor $A$ is defined as $\rho(A)=\max \{|\lambda|: \lambda \text { is an eigenvalue of } A\}$.
A tensor ${C}=\left(c_{i_1} \ldots c_{i_k}\right)$ of order $k$ and dimension $n$ is termed reducible if there exists a nonempty proper index subset $I \subset\{1, \ldots, n\}$ such that
$$
c_{i_1 \cdots i_k}=0 \quad \forall i_1 \in I, \quad \forall i_2, \ldots, i_k \notin I.
$$
If ${C}$ is not reducible, it is referred to as irreducible. A tensor with all non-negative entries is called a non-negative tensor.

Let ${A}$ and ${B}$ be two tensors of order $k$ and dimension $n$. The tensors ${A}$ and ${B}$ are diagonally similar if there exists an invertible diagonal matrix $D$ such that ${B}=D^{-(k-1)} {A} D$, where the tensor-matrix multiplication is introduced in \cite{shao2013general}.

\begin{lemma} [\cite{shao2013general}]\label{lem:eig}
    Suppose that the two tensors ${A}$ and ${B}$ are diagonally similar, namely ${B}=D^{-(k-1)} {A} D$ for some invertible diagonal matrix $D$. Then $x$ is an eigenvector of ${B}$ corresponding to the eigenvalue $\lambda$ if and only if $y=D x$ is an eigenvector of ${A}$ corresponding to the same eigenvalue $\lambda$.
\end{lemma}

We briefly introduce here the results regarding the Perron– Frobenius Theorem of an irreducible Metzler tensor \cite{cui2024metzler}.
The diagonal entries of a tensor are the entries with the same indices, such as $A_{iiiiii}$. Entries with different indices are referred to as off-diagonal entries. In a manner similar to the definition of a Metzler matrix, a Metzler tensor is defined as a tensor whose off-diagonal entries are non-negative \cite{cui2024metzler}. It is straightforward that any Metzler tensor $A$ can be represented as $A = B - s\mathcal{I}$, where $s$ is a scalar, and $B$ is a non-negative tensor.


\begin{lemma}[\cite{cui2024metzler}]
A Metzler tensor $A=B-s\mathcal{I}$ always has an eigenvalue that is real and equal to $\lambda(A)=\lambda(B)-s$, $\lambda(A),\lambda(B)$ is an $H$-eigenvalue of $A,B$ respectively.
\end{lemma}

Next, we show the Perron-Frobnius theorem for an irreducible Metzler tensor.
\begin{lemma}[\cite{cui2024metzler}]\label{thm:perron}
If $B$ is an irreducible nonnegative tensor of order $k$ dimension $n$ with $\rho(B)$ the spectral radius of $B$, then $A=B-s\mathcal{I}$ is an irreducible Metzler tensor of order $k$ dimension $n$ and $\lambda(A)=\rho(B)-s$ is the Perron-$H$-eigenvalue of $A$; furthermore the following hold:
\begin{itemize}
\item[(1)] There is a strictly positive eigenvector $x >\mathbf{0}$ corresponding to $\lambda(A)$.
\item[(2)] If $\lambda$ is an eigenvalue of $A$ with nonnegative eigenvector, then $\lambda=\lambda(A)$. Moreover, the nonnegative eigenvector is unique up to a multiplicative constant.
\end{itemize}
\end{lemma}

In the following paragraph, we provide definitions related to hypergraphs.

The notion of a hypergraph is formally introduced in \cite{gallo1993directed}. In this context, we leverage a collection of tensors to describe the weight information associated with a hypergraph. A weighted and directed hypergraph is defined as a triple \(\mathbf{H} = (\mathcal{V}, \mathcal{E}, A)\). Here, \(\mathcal{V}\) represents a set of vertices, and \(\mathcal{E} = \{E_1, E_2, \ldots, E_n\}\) denotes the set of hyperedges. Each hyperedge can be represented as an ordered pair \(E = (\mathcal{X}, \mathcal{Y})\), consisting of disjoint subsets of vertices, where \(\mathcal{X}\) corresponds to the tail of the hyperedge, and \(\mathcal{Y}\) corresponds to the head. 
In the context of modeling, a directed hyperedge often symbolizes the collective influence of a group of agents on a single agent. Consequently, it is sufficient to work with hyperedges having a single tail \cite{cui2024metzler}. So, from now on, we assume that each hyperedge possesses only one tail but can have one or more heads. This framework is in line with the approach presented in \cite{xie2016spectral} and offers the advantage that a directed uniform hypergraph can be efficiently represented using an adjacency tensor.
To represent the weights associated with hyperedges, we introduce the set of tensors \(A = \{A_2, A_3, \ldots\}\). Here, \(A_2 = [A_{ij}]\) captures the weights of second-order hyperedges, \(A_3 = [A_{ijk}]\) represents the weights of third-order hyperedges, and so on. For example, \(A_{ijkl}\) denotes the weight of the hyperedge where \(i\) serves as the tail, and \(j, k, l\) are the heads. To simplify our exposition, we use the weight notation (e.g., \(A_{\bullet}\)) to refer to the respective hyperedge. If all hyperedges in the hypergraph consist of an equal number of nodes, we classify the hypergraph as uniform. A $k$-uniform hypergraph is a hypergraph, where all hyperedges consist of $k$ nodes. For comprehensive discussions on this, please refer to \cite{chen2021controllability} for insights into undirected uniform hypergraphs and to \cite{xie2016spectral} for a detailed treatment of directed uniform hypergraphs. Conventionally, an undirected $k$-uniform hypergraph can be described using a supersymmetric adjacency tensor of order $k$. A directed $k$-uniform hypergraph can be represented by using an adjacency tensor of order $k$. A uniform undirected hypergraph is connected if its adjacency tensor is irreducible. A uniform directed hypergraph is strongly connected if its adjacency tensor is irreducible. A non-uniform (un)directed hypergraph can be regarded as a composition of several uniform (un)directed hypergraphs. Thus, we can use a set of tensors (of different orders) to represent a non-uniform hypergraph.

\section{Positive Metzler systems with a zero Perron-eigenvalue}
Before we go further into higher-order Laplacian dynamics and consensus-convergence dynamics on hypergraphs, we introduce a series of results regarding a positive Metzler-tensor-based system. Later, we leverage these results to study Laplacian dynamics and design a convergence protocol. We consider a positive Metzler system \cite{cui2024metzler} of $n$ nodes or agents. The dynamical system is given by
\begin{equation}\label{eq:sys1}
    \dot{x}=Ax^{k-1},
\end{equation}
where $A$ is an irreducible Metzler tensor of order $k$ dimension $n$, and $x\in\mathbb{R}^n$ is the state variable. Component-wise, \eqref{eq:sys1} reads as
\begin{equation}
    \dot{x}_i=\sum_{i_2, \ldots, i_k=1}^n A_{i, i_2 \cdots i_m} x_{i_2} \cdots x_{i_k}.
\end{equation}

Firstly, we show that the system \eqref{eq:sys1} is a positive system.
\begin{lemma}
    The positive orthant $\mathbb{R}_+^n$ 
    is positively invariant with respect to the flow of \eqref{eq:sys1}.
\end{lemma}

\begin{proof}
    If $x_i=0$, then $\dot{x}_i=\sum_{i_2, \ldots, i_m\neq i} A_{i, i_2 \cdots i_k} x_{i_2} \cdots x_{i_k}\geq 0.$ 
\end{proof}

Next, we present the following main result of our paper.


\begin{prop}\label{thm:centermanifold}
 Consider the system \eqref{eq:sys1} on a $k$-uniform hypergraph. If the tensor $A$ has a zero Perron-eigenvalue, then \eqref{eq:sys1} has a line of equilibria spanned by the corresponding Perr-eigenvector $x^*$; that is, the line of equilibria can be written as $\alpha x^*$, $\alpha\geq0$. Such a line of equilibria is globally attractive. Furthermore, for any positive initial condition, the corresponding solution of \eqref{eq:sys1} is bounded and does not converge to the origin.
\end{prop}

\begin{proof}
     Suppose that the tensor $A$ has a zero Perron-eigenvalue, and let $x^*$ be the corresponding positive eigenvector. Then, according to equation \eqref{eq:eigenproblem}, we have $A(x^*)^{k-1}=0$. This shows that $x^*$ is a positive equilibrium of system \eqref{eq:sys1}. Now, let $\alpha\geq0$ be a constant. According to \eqref{eq:eigenproblem}, $\alpha x^*$ is also a positive eigenvector corresponding to the zero Perron-eigenvalue. Notice that the origin is always an equilibrium. Since the constant $\alpha\geq 0$ can be arbitrary, $\alpha x^*$ forms a line of equilibria.

     Next, we show that this line of equilibria is globally attractive. Define the Lyapunov-like function $V_m= \max_i \left(\frac{x_i}{x^*_i}\right)^{k-1}$. Let $m=\arg \max_i\left(\frac{x_i}{x^*_i}\right)^{k-1}$. We can see that $V_m\geq 0$. Moreover, $V_m= 0$ if and only if $x=\mathbf{0}$, and $V_m$ is radially unbounded. In addition, the function $V_m$ is Lipschitz continuous in a bounded set because it is a polynomial function.
     
    Next, for any $i$, we have
    \begin{equation}\label{eq:ineq}
        x_i\leq \max_i \left(\frac{x_i}{x^*_i}\right)^{\frac{k-1}{k-1}}x^*_i=V_m^{\frac{1}{k-1}}x^*_i.
    \end{equation}
    This inequality \eqref{eq:ineq} holds as a strict equality only when $\frac{x_i}{x^*_i}$ is $\max_j\left(\frac{x_j}{x^*_j}\right)$.
    
    Let $T_i=(k-1) x_i^{k-2}$ for $i=1,\ldots,n$. Then we get 
    {\small\begin{equation}
    \begin{split}
        \dot{V}_m &= \frac{T_m\dot{x}_m}{(x_m^*)^{k-1}}=\frac{T_m}{(x_m^*)^{k-1}}\left(  \sum_{i_2, \ldots, i_k=1}^n A_{m, i_2 \cdots i_k} x_{i_2} \cdots x_{i_k} \right)\\
        &= \frac{T_m}{(x_m^*)^{k-1}}\Big(  \sum_{i_2, \ldots, i_k\neq m} A_{m, i_2 \cdots i_k} x_{i_2} \cdots x_{i_k} \\
        &+ A_{m,m,\cdots,m}\quad x_m^{k-1}\Big)\\
        & \leq T_m\left( \sum_{i_2, \ldots, i_k\neq m} A_{m, i_2 \cdots i_k} V_m^{\frac{1}{k-1}} x^*_{i_2} \cdots V_m^{\frac{1}{k-1}} x^*_{i_k} \frac{1}{(x_m^*)^{k-1}} \right.\\  
        &+  A_{m,m,\cdots,m}\quad V_m \left.\right)\\
        &= T_mV_m \left( \sum_{i_2, \ldots, i_k\neq m}A_{m, i_2 \cdots i_k} \frac{x^*_{i_2}\cdots x^*_{i_k}}{(x_m^*)^{k-1}}+ A_{m,\cdots ,m}\right)\\
        &= T_mV_m  \lambda(A)=0
    \end{split}
\end{equation}}

Since $\dot{V}_m \leq 0$, then $0\leq V_m \leq V_m(x(0))$. Thus, the solution of the system \eqref{eq:sys1} will converge to the largest invariant set of $M=\{x|\dot{V}_m(x)=0\}$ from lemma \ref{lem:inv} (the index of $m$ may swap). Recall that the inequality \eqref{eq:ineq} holds as an equality when $x_i=x_m$. Then, by further considering that $A$ is irreducible, $\dot{V}_m(x)=0$ implies $\frac{x_i}{x_m}$ is $\max_j\left(\frac{x_j}{x^*_j}\right)$ for all $i$. For any $i,j$, it holds $\frac{x_i}{x^*_i}=\frac{x_j}{x^*_j}$. Therefore, the solution of \eqref{eq:sys1} converges to $\alpha x^*$ with $\alpha \geq 0$.

Since $x_i\leq V_m^{\frac{1}{k-1}}x^*_i$ and $V_m$ is non-increasing, $x_i$ is bounded by $\sup(V_m)^{\frac{1}{k-1}} x^*_i$, which is the last statement.

Now, we show that the solution of the system \eqref{eq:sys1} does not converge to the origin from any positive initial condition. Firstly, we derive an error dynamic by a change of coordinate $y=x-x^*$: $\dot{y}=A(y+x^*)^{k-1}$. Componentwise, it is $\dot{y}_i=\sum_{i_2, \ldots, i_k=1}^n A_{i, i_2 \cdots i_m} (y_{i_2}+x_{i_2}^*) \cdots (y_{i_k}+x_{i_k}^*)$. We continue the proof by showing that the error dynamic is a positive system.

    Clearly, if $y_i=0$ for an arbitrary $i$, $\dot{y_i}=\sum_{i_2, \ldots, i_k=1} A_{i, i_2 \cdots i_m} (y_{i_2}+x_{i_2}^*) \cdots (y_{i_k}+x_{i_k}^*)$. Consider the only negative term is $A_{ii\cdots i}(x_i^*)^{k-1}$. Considering that $\sum_{i_2, \ldots, i_k=1}^n A_{i, i_2 \cdots i_m} x^*_{i_2} \cdots x^*_{i_k}=0$ and $A$ is a Metzler tensor, all the rest terms are non-negative. 

    Thus, $\dot{y_i}\geq 0$. The error dynamic is a positive system. We know that $\alpha x^*$ is also an equilibrium. Thus, from a positive initial condition $x\geq \alpha x^*$, the solution of the system \eqref{eq:sys1} remains non-smaller than $x^*$. Since we can choose $\alpha$ arbitrarily small, one completes the proof. 
\end{proof}

\begin{remark}
    The paper \cite{cui2024metzler} provides the analytical results regarding system \eqref{eq:sys1} but doesn't cover the case when the tensor $A$ has a zero eigenvalue. Our result of Proposition \ref{thm:centermanifold} provides further details under such a case.
\end{remark}

Next, we consider a more complicated scenario. Consider a Metzler positive system on a general non-uniform hypergraph:
\begin{equation}\label{eq:sysnon}
    \dot{x}=A_{k-1} x^{k-1}+A_{k-2} x^{k-2}+\cdots +A_{1} x,
\end{equation}
where all the tensors $A_{k-1},\cdots,A_1$ are irreducible Metzler tensors.


\begin{thm}\label{prop:c2}
Consider the system \eqref{eq:sysnon} and suppose all tensors have a common Perron-eigenvector associated with the zero Perron-eigenvalue. The system has a line of equilibria. The line of equilibria is globally attractive in the positive orthant and the solution of the system \eqref{eq:sysnon} doesn't converge to the origin from any positive initial condition. If the initial condition is finite, the solution remains bounded.
\end{thm}

\begin{proof}
    Suppose that all tensors $A_i$ have a zero Perron-eigenvalue, and let $x^*$ be the corresponding positive eigenvector. Then, according to equation \eqref{eq:eigenproblem}, we have $A_i(x^*)^{k-1}=0$. This shows that $\sum_{i=1}^{k-1} A_i(x^*)^{k-1}=0$ and thus $x^*$ is a positive equilibrium of system \eqref{eq:sysnon}. Now, let $\alpha>0$ be a constant. According to the equation \eqref{eq:eigenproblem}, $\alpha x^*$ is also a positive eigenvector corresponding to the zero Perron-eigenvalue. Since the constant $\alpha\geq 0$ can be arbitrary, thus $\alpha x^*$ forms a line of equilibria.

    Define a Lyapunov function as $V_m= \max_i \left(\frac{x_i}{x^*_i}\right)^{k-1}$. Let $i=\arg \max_k\left(\frac{x_k}{x^*_k}\right)^{k-1}$. Let $T_i=(k-1) x_i^{k-2}$.

    We now obtain
    \begin{equation}
    \begin{split}
        \dot{V}_i &= \frac{T_i\dot{x}_i}{(x^*_i)^{k-1}}=\frac{T_i}{(x^*_i)^{k-1}}\left(  \sum_{i_2, \ldots, i_k=1}^n (A_{k-1})_{ i, i_2 \cdots i_k} x_{i_2} \cdots x_{i_k} \right.\\
        &+\left. \sum_{i_2, \ldots, i_k=1}^n (A_{k-2})_{ i, i_2 \cdots i_{k-1}} x_{i_2} \cdots x_{i_{k-1}} +\cdots\right)\\
        &\leq \frac{T_i}{(x^*_i)^{k-1}} \left(  \sum_{i_2, \ldots, i_k=1}^n (A_{k-1})_{ i, i_2 \cdots i_k} x^*_{i_2} \cdots x^*_{i_k} V_m \right.\\
        &+\left. \sum_{i_2, \ldots, i_k=1}^n (A_{k-2})_{ i, i_2 \cdots i_{k-1}} x^*_{i_2} \cdots x^*_{i_{k-1}} V_m^{\frac{k-2}{k-1}}+\cdots\right)\\
        &= \frac{T_i}{(x^*_i)^{k-1}} \left(\sum_{i=1}^{k-1} V_m^{\frac{i}{k-1}} \lambda(A_{i})(x^*_i)^{i}\right)=0.
    \end{split}
\end{equation}

The rest of the proof is analog to the proof of Proposition \ref{thm:centermanifold}. 
\end{proof}

\section{Laplacian tensors, higher-order Laplacian dynamics, and consensus dynamics on a signless hypergraph}\label{sec:laplacian}
In this section, we propose higher-order Laplacian dynamics by using Laplacian tensors. Then, we further utilize the results from the previous sections to study higher-order Laplacian dynamics.

The convergence problem on a signless $k$-uniform hypergraph can be formulated as the following. Consider the system of integrators:
\begin{equation}\label{eq:integrator}
    \dot{x}=u, \; x\in \mathbb{R}^n, \;u\in \mathbb{R}^n.
\end{equation}
The goal of the convergence problem is to design a distributed feedback control law $u_i= f(x_i,x_{i_1},\cdots, x_{i_{t},\cdots })$, where $x_{i_1},\cdots, x_{i_{t}},\cdots$ are in the hyperedge that contains $i$ ($i$'s neighbor on a hypergraph), so that the system converges to an identical value, i.e. $x_i=x_j=a$ for any $i$ and $j$, where $a$ is a constant. The distributed feedback control law is based on the node itself and its neighbors in the sense of a hypergraph.

When $k=2$, the hypergraph becomes a graph, and the corresponding convergence dynamics can be designed as Laplacian dynamics:
\begin{equation*}
    \dot{x}=-L x,
\end{equation*}
where $L$ is a Laplacian matrix of a graph.

Inspired by this fact, we propose here a higher-order Laplacian dynamic, and we further show that it can be used as the convergence protocol on a signless $k$-uniform hypergraph under some appropriate assumptions. The higher-order Laplacian dynamic reads as 
\begin{equation}\label{eq:consensus}
    \dot{x}=-L x^{k-1},
\end{equation}
where $L$ is a Laplacian tensor of a uniform hypergraph. The higher-order terms $-L x^{k-1}$ capture how agents interact with each other inside a group of $k$ agents (a hyperedge with $k$ agents). There are multiple definitions of the Laplacian tensor. We will introduce them with details later in this section. When $k=2$, all such definitions reduce to the classical Laplacian dynamics on a conventional graph. The interaction function $-L x^{k-1}$ is in a multiplicative form component-wise. Note that the multiplicative form is very suitable to describe the higher-order interactions among agents. For example, the probability of two events that happen simultaneously
is the product of their probabilities \cite{chen2021controllability,cui2023general}. In ecology, the multiplicative form represents the influence of one species on another correlated to a third species \cite{letten2019mechanistic,cui2024analysis}. In addition, the higher order Laplacian dynamic \eqref{eq:consensus} is suitable for example for opinion dynamics. Recall that the componentwise representation of the model \eqref{eq:consensus} is 
\begin{equation*}
    \dot{x}_i=\sum_{i_2, \ldots, i_k=1}^n -L_{i, i_2 \cdots i_m} x_{i_2} \cdots x_{i_k}.
\end{equation*}
The product $x_{i_2} \cdots x_{i_k}$  can be understood as either a joint influence of $i_2,\cdots,i_k$ on $i$'s opinion or the influence of $i_k$ on $i$'s opinion accompanied with some indirect effects due to $i_2,\cdots, i_{k-1}$. At the end, we want to emphasize that the product $L x^{k-1}$ is a vector, which is suitable for the modeling of a vector field. If the power is not $k-1$, the product will be a matrix or a tensor and is potentially not suitable for the modeling purpose. 

In \cite{qi2013h}, the following definition of an adjacency tensor for an undirected uniform hypergraph and a Laplacian tensor is proposed:

\begin{definition}[\cite{qi2013h}]\label{def:1}
   The adjacency tensor ${A}={A}(\mathbf{H})$ of $\mathbf{H}$
   , is a $k$ th order $n$-dimensional symmetric tensor, with ${A}=\left(A_{i_1 i_2 \cdots i_k}\right)$, where $A_{i_1 i_2 \cdots i_k}=\frac{1}{(k-1) !}$ if $\left(i_1, i_2, \ldots, i_k\right) \in E$, and 0 otherwise. For $i \in V$, its degree $d(i)$ is defined as $d(i)=\left|\left\{e_p: i \in e_p \in E\right\}\right|$. We assume that every vertex has at least one edge. Thus, $d(i)>0$ for all $i$. The degree tensor ${D}={D}(\mathbf{H})$ of $\mathbf{H}$, is a $k$-th order $n$-dimensional diagonal tensor, with its $i$-th diagonal entry as $d(i)$. The Laplacian tensor is simply defined by $L=D-A$. 
\end{definition}

The Laplacian tensor as defined by Definition \ref{def:1} has the following properties:
\begin{lemma}[Theorem 4 \cite{qi2013h}]\label{lem:propl1}
     Assume that $k \geq 3$. Zero is the unique $H^{++}$-eigenvalue of $L$ with $H$-eigenvector $\mathbf{1}$, and is the smallest $H$-eigenvalue of $L$.
\end{lemma}

\begin{cor}
    Consider the system \eqref{eq:consensus} on a connected hypergraph with the Laplacian as given in Definition \ref{def:1}. The system converges to an identical value from any positive initial condition, i.e. $\lim_{t\rightarrow \infty} x_i=\alpha>0, \forall i=1,\cdots, n$.
\end{cor}

\begin{proof}
    This is a direct consequence of Lemma \ref{lem:propl1} and Proposition \ref{thm:centermanifold}.
\end{proof}

In the paper \cite{hu2015laplacian}, there is another way of defining an adjacency tensor for an undirected uniform hypergraph and a Laplacian tensor.
\begin{definition} [\cite{hu2015laplacian}]\label{def:2}
Let $\mathbf{H}$ be a $k$-uniform hypergraph with vertex set $[n]=\{1, \ldots, n\}$ and edge set $E$. The normalized adjacency tensor ${A}$, which is a  symmetric nonnegative tensor of order $k$ dimension $n$, is defined as
$$
A_{i_1 i_2 \ldots i_k}:= \begin{cases}\frac{1}{(k-1) !} \prod_{j \in[k]} \frac{1}{\sqrt[k]{d_{i_j}}} & \text { if }\left\{i_1, i_2 \ldots, i_k\right\} \in E, \\ 0 & \text { otherwise. }\end{cases}
$$
The normalized Laplacian tensor ${L}$, which is a symmetric tensor of order $k$ dimension $n$, is defined as
$$
{L}:=D-A,
$$
where $D$ is a diagonal tensor of order $k$ dimension $n$ with the $i$-th diagonal element $D_{i \ldots i}=1$ whenever $d(i)>0$, and zero otherwise.
\end{definition}


\begin{lemma}\label{lem:propl2}
Consider the Definition \ref{def:2}. Assume that $k \geq 3$. Zero is the unique $H^{++}$-eigenvalue of $L$ with $H$-eigenvector $\Tilde{d}=(\sqrt[k]{d_1},\cdots, \sqrt[k]{d_n})^\top$, and is the smallest $H$-eigenvalue of $L$.
\end{lemma}

\begin{proof}
    The statement, zero is the unique $H^{++}$-eigenvalue of $L$ and is the smallest $H$-eigenvalue of $L$, is proven by Corollary 3.2 of \cite{hu2015laplacian}. It is easy to check that $L\Tilde{d}^{k-1}=\mathbf{0}$. Therefore, $\Tilde{d}$ is the corresponding eigenvector.
\end{proof}

\begin{cor}
    Consider the system \eqref{eq:consensus} on a connected hypergraph with the Laplacian as in Definition \ref{def:2}. From any positive initial condition,  the system admits a
    clustering consensus solution based on the degree of the nodes, i.e. $\lim_{t\rightarrow \infty} x_i>0, \forall i=1,\cdots,n$ and $\lim_{t\rightarrow \infty} x_i=x_j, \forall d(i)=d(j)$.
\end{cor}

\begin{proof}
    This is the direct consequence of Lemma \ref{lem:propl2} and Proposition \ref{thm:centermanifold}.
\end{proof}

Note that both Definitions \ref{def:1} and \ref{def:2} require the normalization of the weights and thus only apply to the unweighted hypergraph. Next, we can deal with a more general case of an arbitrary directed weighted hypergraph. We propose a formal definition of a Laplacian tensor for a signless hypergraph as follows:

\begin{definition} \label{def:3}
Let $\mathbf{H}$ be a $k$-uniform hypergraph with vertex set $[n]=\{1, \ldots, n\}$ and edge set $E$. The adjacency tensor ${A}$, which is a nonnegative tensor of order $k$ dimension $n$, is defined as
$A_{i_1 i_2 \ldots i_k}\in\mathbb{R}_+ \text { if }\left\{i_1, i_2 \ldots, i_k\right\} \in E; A_{i_1 i_2 \ldots i_k}= 0 \text { if otherwise.}$
The Laplacian tensor ${L}$, which is a symmetric tensor of order $k$ dimension $n$, is defined as
$$
{L}:=D-A,
$$
where $D$ is a diagonal tensor of order $k$ dimension $n$ with the $i$-th diagonal element $D_{i \ldots i}=\sum_{i_2, \ldots, i_k=1}^n A_{i,i_2 \cdots i_k}$.
\end{definition}

This definition can be regarded as an extension to the definition of a Laplacian matrix of a directed weighted graph.

\begin{lemma}\label{lem:propl3}
     Assume that $k \geq 3$. Zero is the unique $H^{++}$-eigenvalue of $L$ with $H$-eigenvector $\mathbf{1}$, and is the smallest $H$-eigenvalue of $L$.
\end{lemma}

\begin{proof}
    It is straightforward to check that $L(\mathbf{1})^{k-1}=\mathbf{0}$. Since $\mathbf{1}$ is strictly positive, then zero is the unique $H^{++}$-eigenvalue and the smallest $H$-eigenvalue of $L$ according to Theorem \ref{thm:perron}.
\end{proof}

\begin{cor}\label{cor:1}
    Consider the system \eqref{eq:consensus} on a strongly connected signless hypergraph with the Laplacian as in Definition \ref{def:3}. From any positive initial condition, the system converges to an identical value, i.e. $\lim_{t\rightarrow \infty} x_i=\alpha>0, \forall i=1,\cdots, n$.
\end{cor}

\begin{proof}
    This is the direct consequence of Lemma \ref{lem:propl3} and Proposition \ref{thm:centermanifold}.
\end{proof}

The general non-uniform hypergraph up to the leading order $k$ can be regarded as a multilayer network consisting of a $2$-uniform hypergraph up to a $k$-uniform hypergraph. Now, we consider the following system:

\begin{equation}\label{eq:sysnonl}
    \dot{x}=-L_{k-1} x^{k-1}-L_{k-2} x^{k-2}+\cdots -L_{1} x,
\end{equation}
where all the tensors $L_{k-1},\cdots,L_1$ are Laplacian tensors of each layer of uniform sub-hypergraph.

\begin{cor}\label{cor:2}
    Consider the system \eqref{eq:sysnonl} on a non-uniform signless hypergraph with each Laplacian tensor $L_i$ as in Definition \ref{def:3}. Assume that each uniform sub-hypergraph is strongly connected. From any positive initial condition, the system converges to an identical value, i.e. $\lim_{t\rightarrow \infty} x_i=\alpha>0, \forall i=1,\cdots, n$.
\end{cor}

\begin{proof}
    This is the direct consequence of Lemma \ref{lem:propl3} and Theorem \ref{prop:c2}.
\end{proof}

According to our findings, we can draw a conclusion that the collective behavior on a signless hypergraph is similar to that on a conventional signless graph in the positive orthant.

\section{Higher-order Laplacian dynamic and bipartite consensus on a signed hypergraph}

In this section, we further study the higher-order Laplacian dynamics on signed hypergraphs and we show that they can be used to analyze the collective behavior on higher-order networks.
Firstly, one of the general collective behavior that we can observe from a graph model (such as an Altafini model \cite{altafini2012consensus}) is usually bipartite consensus, i.e. $\lim_{t\rightarrow \infty} |x_i|=\alpha>0, \forall i=1,\cdots, n$.

Inspired by \cite{altafini2012consensus}, we propose the Laplacian tensor for a signed directed uniform hypergraph as follows:

\begin{definition} \label{def:4}
Let $\mathbf{H}$ be a $k$-uniform hypergraph with vertex set $[n]=\{1, \ldots, n\}$ and edge set $E$. The adjacency tensor ${A}$, which is a  nonnegative tensor of order $k$ dimension $n$, is defined as
$A_{i_1 i_2 \ldots i_k}\in\mathbb{R} \text { if }\left\{i_1, i_2 \ldots, i_k\right\} \in E; A_{i_1 i_2 \ldots i_k}= 0 \text { otherwise.}$
The Laplacian tensor ${L}$, which is a symmetric tensor of order $k$ dimension $n$, is defined as
$$
{L}:=D-A,
$$
where $D$ is a diagonal tensor of order $k$ dimension $n$ with the $i$-th diagonal element $D_{i \ldots i}=\sum_{i_2, \ldots, i_k=1}^n |A_{i,i_2 \cdots i_k}|$.
\end{definition}


We consider the system \eqref{eq:consensus} with the Definition \ref{def:4} of the Laplacian. Firstly, we perform a gauge transformation $z=Gx$ , where $G$ is a diagonal matrix with diagonal entries being either $-1$ or $1$, for the system \eqref{eq:consensus}. Notice that $G=G^\top=G^{-1}$.
Via the gauge transformation, we have 
\begin{equation*}
    \dot{z}_i=-G_{ii}^{-1}\sum_{i_2, \ldots, i_k=1}^n L_{i, i_2 \cdots i_m} G_{i_2 i_2}z_{i_2} \cdots G_{i_k i_k}z_{i_k}.
\end{equation*}

Then, the system after the gauge transformation can be written as
\begin{equation}\label{eq:consensusafg}
    \dot{x}=-L_D x^{k-1},
\end{equation}
where $L_D$ is defined as $(L_D)_{i, i_2 \cdots i_m}=G_{ii}^{-1} L_{i, i_2 \cdots i_m} G_{i_2 i_2} \cdots G_{i_k i_k}$.

Recall that if a tensor $A$ is diagonally similar to a tensor $B$, then there exists an invertible diagonal matrix $P$ such that ${B}=P^{1-k} {A} P$ \cite{shao2013general}. Then
$$
\begin{aligned}
B_{i_1 i_2 \ldots i_k} & =\left(P^{1-k} {A} P\right)_{i_1 i_2 \ldots i_k} \\
& =A_{i_1 i_2 \ldots i_k} P_{i_1 i_1}^{1-k} P_{i_2 i_2} \ldots P_{i_k i_k} .
\end{aligned}
$$

Thus, if $k$ is even, $L_D$ and $L$ are diagonally similar and have the same eigenvalues.

Similar to the convergence dynamics on a signed graph (Altafini model \cite{altafini2012consensus}), we now introduce the notion of the structural balance of a signed uniform hypergraph.

\begin{definition}\label{def:sb}
    A signed $k$-uniform hypergraph $\mathbf{H}(A)$ is said to be \emph{structurally balanced} if it admits a bipartition of the nodes $\mathcal{V}_1, \mathcal{V}_2$, with $\mathcal{V}_1 \cup \mathcal{V}_2=\mathcal{V},\, \mathcal{V}_1 \cap \mathcal{V}_2=\emptyset$ and with $\sigma_i=1, \forall i\in \mathcal{V}_1,\, \sigma_j=-1, \forall j\in \mathcal{V}_2$ such that $\sgn(A_{i_1,\cdots,i_k})=\sigma_{i_1}\sigma_{i_2}\cdots\sigma_{i_k}$. The signed $k$-uniform hypergraph $\mathbf{H}(A)$ is said to be \emph{structurally unbalanced} otherwise. We call the choice of $\mathcal{V}_1, \mathcal{V}_2$ the faction formation. The vector $\sigma$ is the vector of faction formation.
\end{definition}

Then, we have the following result of bipartite consensus.



\begin{prop}\label{thm:bipartite}
    Consider the system \eqref{eq:consensus} on a signed $k$-uniform hypergraph $\mathbf{H}$ with the Laplacian given by Definition \ref{def:4}. Assume that $k$ is even and let $G=\textbf{Diag}(\sigma)$, where $\sigma$ is a vector of faction formation. If $\mathbf{H}$ is structurally balanced with vector of faction formation $\sigma$ and strongly connected, then from any initial condition $x$ such that $Gx\in \mathbb{R}_{++}^n$, the system admits a bipartite consensus solution, i.e. $\lim_{t\rightarrow \infty} |x_i|=\alpha>0, \forall i=1,\cdots, n$ and in addition $\text{Sign}(x_i(t))=\text{Sign}(x_i(0))$, where $\text{Sign}(y)$ denotes the sign of $y$.
    \end{prop}

\begin{proof}
    Let the gauge transformation be consistent with the faction formation of the balanced structure of the hypergraph, i.e, $G=\textbf{Diag}(\sigma)$ where $\sigma$ is the vector of faction formation given in Definition \ref{def:sb}. From Definition \ref{def:sb}, the hypergraph admits a bipartition of the nodes $\mathcal{V}_1, \mathcal{V}_2, \mathcal{V}_1 \cup \mathcal{V}_2=$ $\mathcal{V}, \mathcal{V}_1 \cap \mathcal{V}_2=0$ and let $G_{ii}=\sigma_i=1, \forall i\in \mathcal{V}_1, G_{jj}=\sigma_j=-1, \forall j\in \mathcal{V}_2$. Thus, $(L_D)_{i, i_2 \cdots i_m}=G_{ii}^{-1} L_{i, i_2 \cdots i_m} G_{i_2 i_2} \cdots G_{i_k i_k}= -G_{ii}^{-1} A_{i, i_2 \cdots i_m} G_{i_2 i_2} \cdots G_{i_k i_k}$ for $i_2 \cdots i_m$ not equal to $i$ at the same time. By Definition \ref{def:sb}, $\sgn(A_{i_1,\cdots,i_k})=G_{i_1i_1}G_{i_2i_2}\cdots G_{i_ki_k}$. This implies that $-L_D$ is an irreducible Metzler tensor. Since the hypergraph is structurally balanced with the vector of faction formation $\sigma$, the system after the gauge transformation \eqref{eq:consensusafg} becomes a system on a signless uniform hypergraph. From Corollary \ref{cor:1}, we obtain that the system \eqref{eq:consensusafg} converges to an identical value from any positive initial condition and thus the original system \eqref{eq:consensus} reach bipartite consensus from any initial condition $Gx>0$. 
\end{proof}

\begin{remark}
    Notice that one faction could be empty, and this will lead to a consensus case. The case of a consensus convergence behavior is a special case of bipartite consensus. It is clear that when $k=2$, the Definition \ref{def:sb} is the Definition of structural balance on a signed graph, and in such a case, the model \eqref{eq:consensus} becomes the famous Altafini model \cite{altafini2012consensus} on a signed graph.
\end{remark}

Now, we can further consider the collective behavior on a signed non-uniform hypergraph.

\begin{thm}\label{cor:3}
    Consider the system \eqref{eq:sysnonl} on a signed non-uniform hypergraph $\mathbf{H}$ with the Laplacian of each uniform sub-hypergraph given by Definition \ref{def:4}. From any initial condition $Gx\in \mathbb{R}_{++}^n$, if each uniform sub-hypergraph $\mathbf{H}$ is of even order, structurally balanced with respect to the same faction formation, and strongly connected, the system admits a bipartite consensus solution, i.e. $\lim_{t\rightarrow \infty} |x_i|=\alpha>0, \forall i=1,\cdots, n$.
\end{thm}

\begin{proof}
We again perform gauge transformation to the system \eqref{eq:sysnonl}. We obtain that 
\begin{equation*}
    \dot{x}=-(L_D)_{k-1} x^{k-1}-(L_D)_{k-2} x^{k-2}+\cdots -(L_D)_{1} x,
\end{equation*}
where $[(L_D)_{m}]_{i, i_2 \cdots i_m}=G_{ii}^{-1} (L_m)_{i, i_2 \cdots i_m} G_{i_2 i_2} \cdots G_{i_k i_k}.$

    The rest of the proof is analog to the proof of Proposition \ref{thm:bipartite} by further using Theorem \ref{prop:c2}.
\end{proof}


In case the uniform hypergraph is structurally unbalanced, the system behavior is challenging to analyze, because the system is no longer a Metzler system (after a gauge transformation). However, \cite{chen2022explicit} suggests if the tensor is orthogonally decomposable (odeco), then the solution of the system can be solved explicitly. This implies that the solution may converge to zero or a fixed value or diverge. However, the main drawback of the technique is that not every tensor is odeco. Hence, such an approach can be only applicable to a special class of the system (odeco homogeneous polynomial system). However, a non-odeco homogeneous polynomial system may still be approximated by an odeco homogeneous polynomial system \cite{chen2022explicit}. The remaining details of a structurally unbalanced case generally remain a topic for future work. So far, we have investigated the collective behavior of higher-order Laplacian dynamics given that the initial condition is located in a desired orthant $Gx\in \mathbb{R}_{++}^n$. The system behavior outside this orthant is usually similar to the structurally unbalanced case. Thus, we see that the bipartite consensus may break outside the orthant of interest. For the case of non-uniform hypergraphs, the system behaviors of the unbalanced case or beyond the desired orthant are even more challenging because of the high nonlinearity.

However, the condition on the initial condition has a clear physical meaning in the context of opinion dynamics. The condition on the initial condition $Gx\in \mathbb{R}_{++}^n$ implies that all agents from the same faction believe a common opinion that the faction believes in (how strongly they believe the common opinion is not important). In other words, there is no agent from a faction who holds a different opinion from the faction it is in. Otherwise, it may break the bipartite consensus. In reality, since the agent no longer believes in what its faction believes in, it may change the faction. This leads to a new structurally balanced case with a proper initial condition. This finally brings a new bipartite consensus state.

Thus, compared with the classical Altafini model on a graph \cite{altafini2012consensus}, the bipartite consensus of the higher-order Laplacian dynamics \eqref{eq:consensusafg} is less robust. The reason is not difficult to see. While the classical Altafini model is a linear model on a graph, the higher-order Laplacian dynamics is a homogeneous (non-homogeneous) model on a uniform (non-uniform) hypergraph, and from its nonlinear structure, it naturally induces a more complicated system behavior. However, we successfully found out that the desired collective behavior of bipartite consensus will remain for an appropriate orthant and structurally balanced network setting. The results will help us to understand the similarities and differences between collective behavior on a graph and a hypergraph.

\section{Interaction functions other than polynomials}
In this section, we adopt a non-polynomial interaction function into the framework of the higher-order Laplacian dynamic. 
Firstly, we discuss the system in the form of 
\begin{equation}\label{eq:abnon}
\dot{x}_i=\sum_{i_2,\cdots,i_k} -L_{i,i_2,\cdots,i_k} f(x_{i_2})\cdots f(x_{i_k}).
\end{equation}

We make the following assumptions on the function $f(x_i)$.
\begin{assumption}\label{ass:1}
    The function $f(x_i)$ is strictly increasing, continuous and satisfies that $f(x_i)>0$ if $x_i>0$.
\end{assumption}

\begin{lemma}
    The positive orthant $\mathbb{R}_+^n$ is positively invariant with respect to the flow of \eqref{eq:abnon}.
\end{lemma}

\begin{proof}
    If $x_i=0$, then $\Dot{x}_i\geq 0$ due to the assumption 1.
\end{proof}

Next, we show that the system \eqref{eq:abnon} has a similar collective behavior in a desired orthant.

\begin{cor}
    Consider the system \eqref{eq:abnon} under Assumption \ref{ass:1} on a strongly connected signless hypergraph with the Definition \ref{def:3} of Laplacian. From any positive initial condition, the system converges to an identical value, i.e. $\lim_{t\rightarrow \infty} x_i=\alpha>0, \forall i=1,\cdots, n$.
\end{cor}

\begin{proof}
    The proof is similar to the proof of Proposition \ref{thm:centermanifold}. Now one takes $V_m=\max_i (\frac{f_i(x_i)}{x^*_i})^{k-1}$ as the Lyapunov function. Also, note that a continuous function on a compact set is uniformly continuous on that set. Following the same proof of Proposition \ref{thm:centermanifold}, we finally know that the solution of \eqref{eq:abnon} converges to the largest invariance set where $f(x_i)=f(x_j)$ for any $i$ and $j$. Because $f(x_i)$ is increasing, this leads to $x_i=x_j$ for any $i$ and $j$.

\end{proof}

The bipartite consensus can be achieved under a further assumption on the interaction function.

\begin{assumption}\label{ass:2}
    The function $f(x_i)$ is an odd function.
\end{assumption}

\begin{cor}\label{cor:np}
Consider the system \eqref{eq:abnon} under Assumptions \ref{ass:1} and \ref{ass:2} on a signed $k$-uniform hypergraph $\mathbf{H}$ with the Definition \ref{def:4} of Laplacian. Assume that $k$ is even. From any initial condition $Gx\in \mathbb{R}_{++}^n$, if the hypergraph $\mathbf{H}$ is structurally balanced and strongly connected, the system admits a bipartite consensus solution, i.e. $\lim_{t\rightarrow \infty} |x_i|=\alpha>0, \forall i=1,\cdots, n$.
\end{cor}

\begin{proof}
    Notice the fact that $f(x_i)$ is an odd function, which means that $f(-x_i)=-f(x_i)$. Thus, all the techniques in the proof of Proposition \ref{thm:bipartite} apply.
\end{proof}

\begin{remark}
    In many related works of nonlinear consensus dynamics \cite{saber2003consensus,sclosa2023graph}, they consider a protocol in the form of 
\begin{equation}\label{eq:nonlinear}
    \Dot{x}_i=\sum_{j\in N_i} \phi_{ij}(x_j-x_i),
\end{equation}
where $N_i$ denotes the set of agent $j$ on a graph and $\phi$ is symmetric such that $\phi_{ij}=\phi_{ji}$. For example, if we assume $\phi_{ij}(x_j-x_i)=(x_j-x_i)^k$ for any $i$ and $j$, then the whole system \eqref{eq:nonlinear} is also a homogeneous polynomial system. However, we can check that the system \eqref{eq:nonlinear} generally can not be represented as a Metzler tensor, which makes \eqref{eq:nonlinear} a totally different system from the dynamics studied in this paper. The main difference is that \eqref{eq:nonlinear} has a conserved quantity, which is the mean value of its initial condition, while the mean value generally is not a conserved quantity of the dynamics in this paper. Another example of conserved quantity is the consensus dynamics studied in \cite{bonettononlinear}. Again, we emphasize that this model \cite{bonettononlinear} is generally not in a Metzler-tensor-based structure.
\end{remark}

\section{Numerical Examples}
In this section, we use a simulation scenario of 4 agents to illustrate our theoretical results.

Firstly, we consider the case in line with Corollary \ref{cor:1}. We set the adjacency tensor as an all-one tensor of order $4$ dimension $4$, which corresponds to a $4$-uniform hypergraph with $4$ nodes. Figure \ref{fig:consensus} shows that all agents finally converge to an identical value from a positive random initial condition. 


Secondly, we consider the case in line with Corollary \ref{cor:2}. We consider a non-uniform hypergraph, which contains a $3$-uniform sub-hypergraph and a sub-graph. For the $3$-uniform sub-hypergraph, the adjacency tensor $A$ is an all-one tensor of order $3$ dimension $4$ except $A_{231}=A_{232}=A_{233}=A_{234}=2$ and $A_{241}=A_{242}=A_{243}=A_{244}=3$. The adjacency matrix of the sub-graph is an all-one matrix in $\mathbb{R}^{n\times n}$. In this way, the hypergraph is a signless network. Figure \ref{fig:consensus2} shows that all agents finally converge to an identical value from a positive random initial condition. 


Thirdly, we consider the case in line with Proposition \ref{thm:bipartite}. Let $\sigma=(1,1,-1,-1)^\top$. We set the adjacency tensor $A$ of order $4$ dimension $4$ such that $\sgn(A_{i_1,i_2,i_3,i_4})=\sigma_{i_1}\sigma_{i_2}\sigma_{i_3}\sigma_{i_4}$. It corresponds to a $4$-uniform hypergraph with $4$ nodes. Figure \ref{fig:biconsensus} shows that all agents finally reach bipartite consensus from a random initial condition satisfying $\Dg(\sigma)x\in \mathbb{R}^n_{++}$. The simulation result is consistent with Proposition \ref{thm:bipartite}.


Then, we consider the case in line with Theorem \ref{cor:3}. We consider a non-uniform hypergraph, which contains a $4$-uniform sub-hypergraph and a sub-graph. Let $\sigma=(1,1,-1,-1)^\top$. We set the adjacency tensor $A$ of order $4$ dimension $4$ such that $\sgn(A_{i_1,i_2,i_3,i_4})=\sigma_{i_1}\sigma_{i_2}\sigma_{i_3}\sigma_{i_4}$. It corresponds to a $4$-uniform sub-hypergraph with $4$ nodes. Furthermore, we set the adjacency matrix $B$ as $\sgn(B_{i_1,i_2})=\sigma_{i_1}\sigma_{i_2}$, which represent a sub-graph with $4$ nodes. Figure \ref{fig:biconsensus2} shows that all agents finally reach bipartite consensus from a random initial condition satisfying $\Dg(\sigma)x\in \mathbb{R}^n_{++}$. The simulation result is consistent with Theorem \ref{cor:3}.

Finally, we consider the non-polynomial case in line with Corollary \ref{cor:np}. We consider a non-uniform hypergraph, which contains a $4$-uniform sub-hypergraph and a sub-graph. Let $\sigma=(1,1,-1,1)^\top$. We set the adjacency tensor $A$ of order $4$ dimension $4$ such that $\sgn(A_{i_1,i_2,i_3,i_4})=\sigma_{i_1}\sigma_{i_2}\sigma_{i_3}\sigma_{i_4}$. It corresponds to a $4$-uniform sub-hypergraph with $4$ nodes. Furthermore, we set the adjacency matrix $B$ as $\sgn(B_{i_1,i_2})=\sigma_{i_1}\sigma_{i_2}$, which represents a sub-graph with $4$ nodes. We choose a non-polynomial interaction function as $f(x_i)=\arctan(x_i)$. Figure \ref{fig:biconsensusnp} shows that all agents finally reach  bipartite consensus from a random initial condition satisfying $\Dg(\sigma)x\in \mathbb{R}^n_{++}$. The simulation result is consistent with Corollary \ref{cor:np}.

\begin{figure*}
\begin{subfigure}[t]{0.24\linewidth}
    \centering
\includegraphics[height=3.5cm]{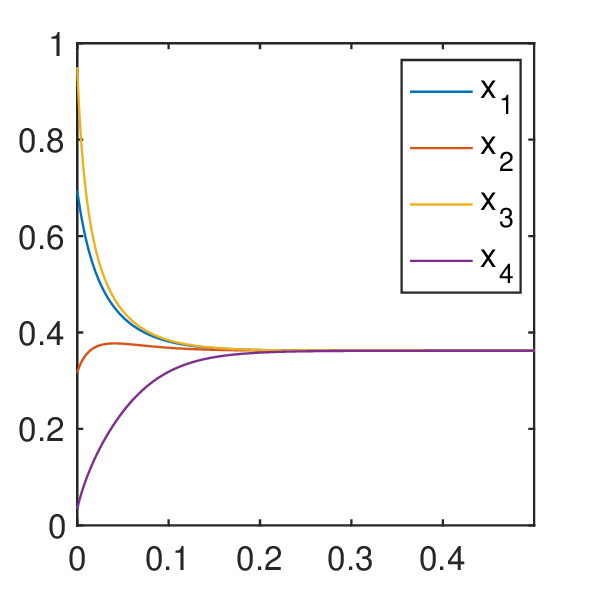}
\captionsetup{width=.95\textwidth}
    \caption{}
    \label{fig:consensus}
  \end{subfigure}
  \begin{subfigure}[t]{0.24\linewidth}
    \centering
\includegraphics[height=3.5cm]{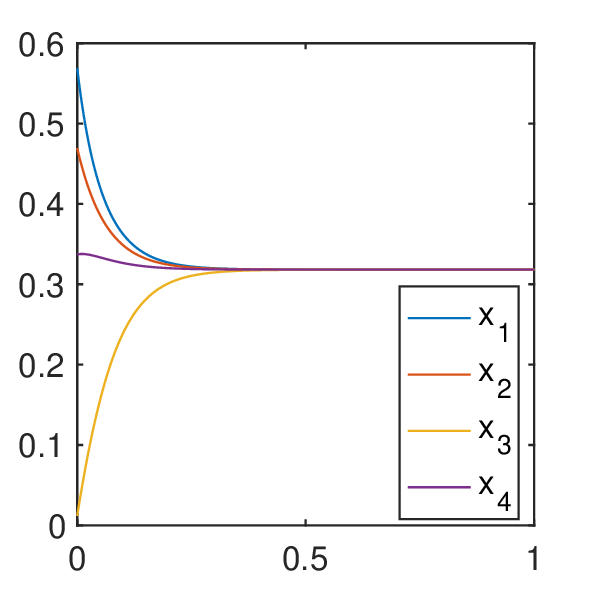}
\captionsetup{width=.95\textwidth}
    \caption{}
    \label{fig:consensus2}
  \end{subfigure}
  \begin{subfigure}[t]{0.24\linewidth}
    \centering
\includegraphics[height=3.5cm]{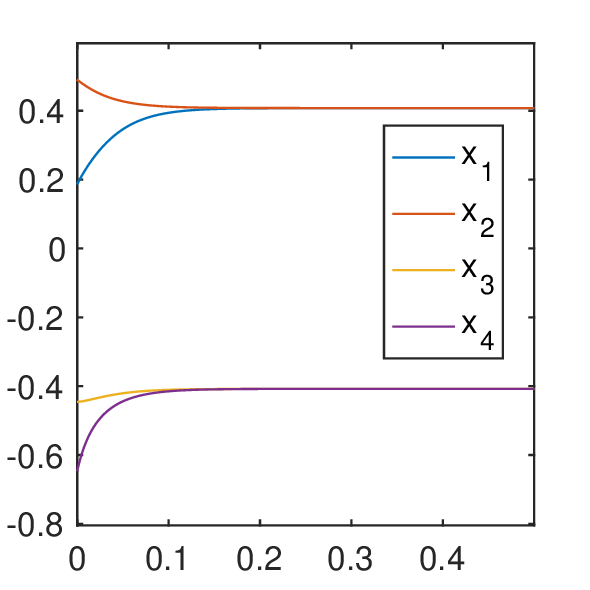}
\captionsetup{width=.95\textwidth}
    \caption{}
    \label{fig:biconsensus}
\end{subfigure}
\begin{subfigure}[t]{0.24\linewidth}
    \centering
\includegraphics[height=3.5cm]{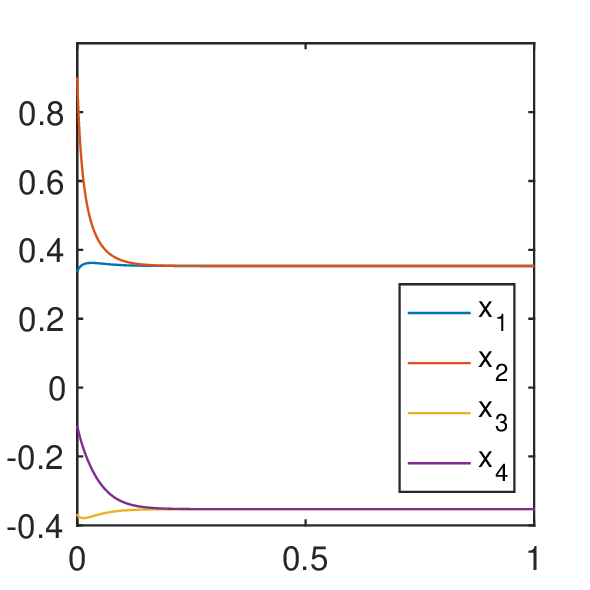}
\captionsetup{width=.95\textwidth}
    \caption{}
    \label{fig:biconsensus2}
\end{subfigure}
\caption{ The horizontal axis is the time $t$.
(a) The consensus example on a $4$-uniform hypergraph.
(b) The consensus example on a non-uniform hypergraph containing both a $3$-uniform sub-hypergraph and a sub-graph.
(c) The bipartite consensus example on a $4$-uniform hypergraph. (d)  The bipartite consensus example on a non-uniform hypergraph containing both a $4$-uniform sub-hypergraph and a sub-graph. 
}
\end{figure*}

\begin{figure}
        \centering
        \includegraphics[height=4cm]{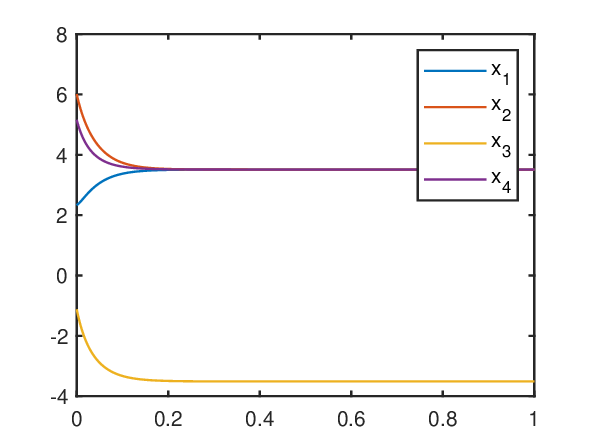}
        \caption{The bipartite consensus example of a non-polynomial interaction on a non-uniform hypergraph containing both a $4$-uniform sub-hypergraph and a sub-graph. The horizontal axis is the time $t$. }
       \label{fig:biconsensusnp}
\end{figure}

\section{Conclusions}

Here, we summarize the paper as follows. In this paper, we consider collective behavior and design convergence dynamics on a signless hypergraph by exploiting a Metzler-tensor-based homogeneous polynomial system. The proof techniques are novel, which include the Perron–Frobenius Theorem of an irreducible Metzler tensor and an extended invariance principle. We further give a formal definition of a Laplacian tensor of a signless directed uniform hypergraph. The result reveals that the collective behavior on a hypergraph is generally similar to that on a conventional graph in a positive orthant, while the behavior may be different outside the positive orthant. 
 Then, we consider the case on a signed hypergraph. We give a formal definition of a Laplacian tensor of a signed directed uniform hypergraph. We further extend the concept of structural balance on a hypergraph. With the help of a gauge transformation and similarities of tensors, we achieve a result of bipartite consensus. We finally find out that the collective behavior on a hypergraph is still similar to that on a conventional graph in a desired orthant determined by the balanced structure of the hypergraph, while the behavior may be different outside the orthant. All these analytical results will help us to understand the collective behavior on a higher-order network with cooperative and antagonistic interactions. In addition, we also take the non-polynomial interaction function into account and obtain some similar results with the polynomial case. Finally, we use some numerical examples to highlight our contributions.

 As potential future works, we may further study the system behavior of an unbalanced case on a signed hypergraph. This may require one to consider and utilize further tensor properties (not necessarily for a Metzler tensor) to go on the study.

%
\bibliographystyle{IEEEtran}
\bibliography{bib}

\begin{appendix}

\section{Extended invariance principle}
Here, we introduce an extended invariance principle \cite{alvarez2000invariance}. We consider discontinuous dynamic systems governed by differential equations of the form
$$
\dot{x}=f(x),
$$
where $f: \mathbb{R}^n \rightarrow \mathbb{R}^n$ is a piecewise continuous function that undergoes discontinuities on a set $N$ of measure zero. 

\begin{lemma}[\cite{alvarez2000invariance}]\label{lem:inv}
    Suppose that there exists a positive definite, Lipschitz-continuous function $V(x)$ such that
\begin{equation}\label{eq:cond}
\frac{d}{d t} V(x(t))=\left.\frac{d}{d h} V(x(t)+h \dot{x}(t))\right|_{h=0} \leqslant 0
\end{equation}
almost everywhere. Let $M$ be the largest invariant subset of the manifold where the strict equality (5) holds, and let $V(x) \rightarrow \infty$ as $\operatorname{dist}(x, M) \rightarrow \infty$. Then all the trajectories $x(t)$ of $(1)$ converge to $M$, that is,
$$
\lim _{t \rightarrow \infty} \operatorname{dist}(x(t), M)=0.
$$\end{lemma}

\begin{lemma}[\cite{alvarez2000invariance}]
     Condition \eqref{eq:cond} of Theorem 1 is fulfilled if $\frac{d}{d t} V(x(t))=\left.\frac{d}{d h} V(x(t)+h \dot{x}(t))\right|_{h=0}$ is nonpositive at the points of the set $N_V$ where the gradient $\nabla V$ of the function $V(x)$ does not exist, and in the continuity domains of the function $f(x)$ where (4) is expressed in the standard form
$$
\frac{d}{d t} V(x)=\nabla V(x) \cdot f(x), \quad x \in \mathbb{R}^n \backslash\left(N \cup N_V\right)
$$
\end{lemma}

Although the Theorem is designed for a discontinuous dynamical system, we can still use it for a continuous dynamical system if we want to adopt a not continuously differentiable Lyapunov function.
\end{appendix}
\end{document}